\documentclass[10pt,A4paper,conference]{IEEEtran}
\usepackage{amsmath}
\usepackage{amssymb}
\usepackage{upref}
\usepackage{amsfonts}
\usepackage{graphicx}
\usepackage{subfigure}
\usepackage{verbatim}

\newtheorem{defn}{Definition}
\newtheorem{theorem}{Theorem}

\newtheorem{lemma}{Lemma}

\newtheorem{cor}{Corollary}

\newcommand{\cA}{{\cal A}}

\newcommand{\cG}{{\cal G}}
\newcommand{\cH}{{\cal H}}

\newcommand{\bfx}{{\boldsymbol x}}
\newcommand{\bfy}{{\boldsymbol y}}

\newcommand{\al}{\alpha}
\newcommand{\bet}{\beta}

\newcommand{\Ftk}{\smash{\mathbb{F}_{\!2}^{\hspace{1pt}k}}}

\newcommand{\C}{\mathbb{C}}

 \DeclareMathOperator{\dist}{d_H}
\newcommand{\E}{\sf E}

\begin{document}
\bibliographystyle{IEEEtran}
\title{Multidimensional Flash Codes}
\author{\authorblockN{Eitan Yaakobi, Alexander Vardy, Paul H. Siegel, and Jack K. Wolf}
\authorblockA{University of California, San Diego\\ La Jolla, CA $92093-0401$, USA \\ Emails: eyaakobi@ucsd.edu, avardy@ucsd.edu, psiegel@ucsd.edu, jwolf@ucsd.edu}}
\maketitle
\begin{abstract}
Flash memory is a non-volatile computer memory comprised of blocks
of cells, wherein each cell can take on $q$ different levels
corresponding to the number of electrons it contains. Increasing
the cell level is easy; however, reducing a cell level forces all
the other cells in the same block to be erased. This erasing
operation is undesirable and therefore has to be used as
infrequently as possible. We consider the problem of designing
codes for this purpose, where $k$ bits are stored using a block of
$n$ cells with $q$ levels each. The goal is to maximize the number
of bit writes before an erase operation is required. We present an
efficient construction of codes that can store an arbitrary number
of bits. Our construction can be viewed as an extension to
multiple dimensions of the earlier work of Jiang and Bruck, where
single-dimensional codes that can store only $2$ bits were
proposed.
\end{abstract}

\section{Introduction}\label{sec:introduction}
Flash memories are, by far, the most important type of
non-volatile computer memory in use today. They are employed
widely in mobile, embedded, and mass-storage applications, and the
growth in this sector continues at a staggering pace.

A flash memory consists of an array of floating-gate \emph{cells},
organized into \emph{blocks} (a typical block comprises $2^{17}$
to $2^{20}$ cells). Hot-electron injection~\cite{KS67} is used to
inject electrons into a cell, where they become trapped. The
Fowler-Nordheim tunneling~\cite{vZ97} mechanism (field emission)
can~be used to remove electrons from an entire block of cells,
thereby discharging them. The level or ``state'' of a cell is a
function of the amount of charge (electrons) trapped within it.
Historically, flash cells have been designed to store only two
values (one bit); however, \emph{multilevel flash cells} are
actively being developed and are already in use in some
devices~\cite{CGOZ99,GT05}. In multilevel flash cells, voltage is
quantized to $q$ discrete threshold values, say
$0,1,\ldots,q{-}1$. The parameter $q$ can range from $q=2$ (the
conventional two-state case) up to $q = 256$.

The most conspicuous property of flash storage is its inherent
asymmetry between cell programming (charge placement) and cell
erasing (charge removal). While adding charge to a single cell is
a fast and simple operation, removing charge from a cell is very
difficult. In fact, flash memories \emph{do not allow} a single
cell to be erased --- rather \emph{only entire blocks} (comprising
up to $2^{20}$ cells) can be erased. Thus, a single-cell erase
operation requires the cumbersome process of copying an entire
block to a temporary location, erasing it, and then re-programming
all the cells except one. Moreover, since over-programming
(raising the charge of a cell above its intended level) can only
be corrected by a block erasure, in practice a conservative
procedure is used for programming a cell. Charge is injected into
the cell over numerous rounds; after every round, the charge level
is measured and the next-round injection is configured, so that
the charge gradually approaches its desired level. All this is
extremely costly in time and energy.

Codes designed to address this problem were first introduced
in~\cite{Jiang-Allerton,J07,JBB07}, and are called \emph{floating
codes}. Here we address these codes slightly differently under the
name of \emph{flash codes}. Flash codes are a sweeping
generalization of write-once memory codes~\cite{CGM86,FS84,RS82},
designed to maximize the number of times information can be
rewritten before block erasures are required. In a nutshell, the
idea is to use $n$ $q$-level cells to store $k < n \log_2 \!q$
bits, thereby storing less bits than possible (the rate of the
flash code is $k/(n \log_2\! q)$\,). The bits are represented in a
clever way to guarantee that every sequence of up to $t$ writes
(of a single bit) does not lead to any of the $n$ cells exceeding
its maximum value $q-1$. Recently, several more papers have
appeared
\cite{BJB07,CSBB07,ER99,GCKT03,GLR03,JB08,JMSB08,JSB08,MLF08} that
discuss coding techniques for this model of flash memories.

Let us begin by giving a precise definition of \emph{flash codes}.
An insightful way to do so is in terms of a pair of graphs and a
pair of mappings between these graphs.\ The first graph is the
familiar hypercube $\cH_k$.\ The vertices of $\cH_k$, called the
\emph{variable vectors}, are the $2^k$ binary~vec\-tors of length
$k$, with two such vertices $\al,\beta \in\Ftk$ being adjacent iff
$\dist(\al,\beta) = 1$ ($\dist(\al,\beta)$ denotes the Hamming
distance between $\al$ and $\beta$). This graph constitutes the
state transition diagram of the $k$ information bits. A single-bit
write operation corresponds to the traversal of an edge in
$\cH_k$, and a sequence of $t$ writes is a \emph{walk of length
$t$}~in~$\cH_k$. To describe the second graph, set $\cA_q =
\{0,1,\dots,q-1\}$, and think of $\cA_q$ as a subset of the
integers. Now consider the \emph{directed} graph \smash{$\cG_n$}
whose vertices are the $q^n$ vectors of length $n$ over
\smash{$\cA_q$}, and are called the \emph{cell state vectors}.
There is a directed edge from \smash{$\bfx \in \cA_q^n$} to
\smash{$\bfy \in \cA_q^n$} in $\cG_n$ iff $\dist(\bfx,\bfy) = 1$
and in the single position $i$ where $\bfx$ and $\bfy$ differ, we
have $y_i = x_i + 1$. The graph $\cG_n$ is the state transition
diagram of $n$ flash memory cells. Observe that there is a path
from $\bfx$ to $\bfy$ in $\cG_n$ iff $y_i \,{\ge}\, x_i$ for all
$i = 1,2,\dots,n$,~which~reflects the condition that the charge of
memory cells can only increase. The graphs $\cH_k$ and $\cG_n$ are
illustrated in Figure~\ref{fig:1 state transition diagrams} for
the case $k=3$, $n=2$, and $q = 8$.
\begin{figure}
\begin{center}
\vspace{-1.7in} \hspace*{0.5ex}
\raisebox{0.30in}{\includegraphics[totalheight=0.35\textheight]{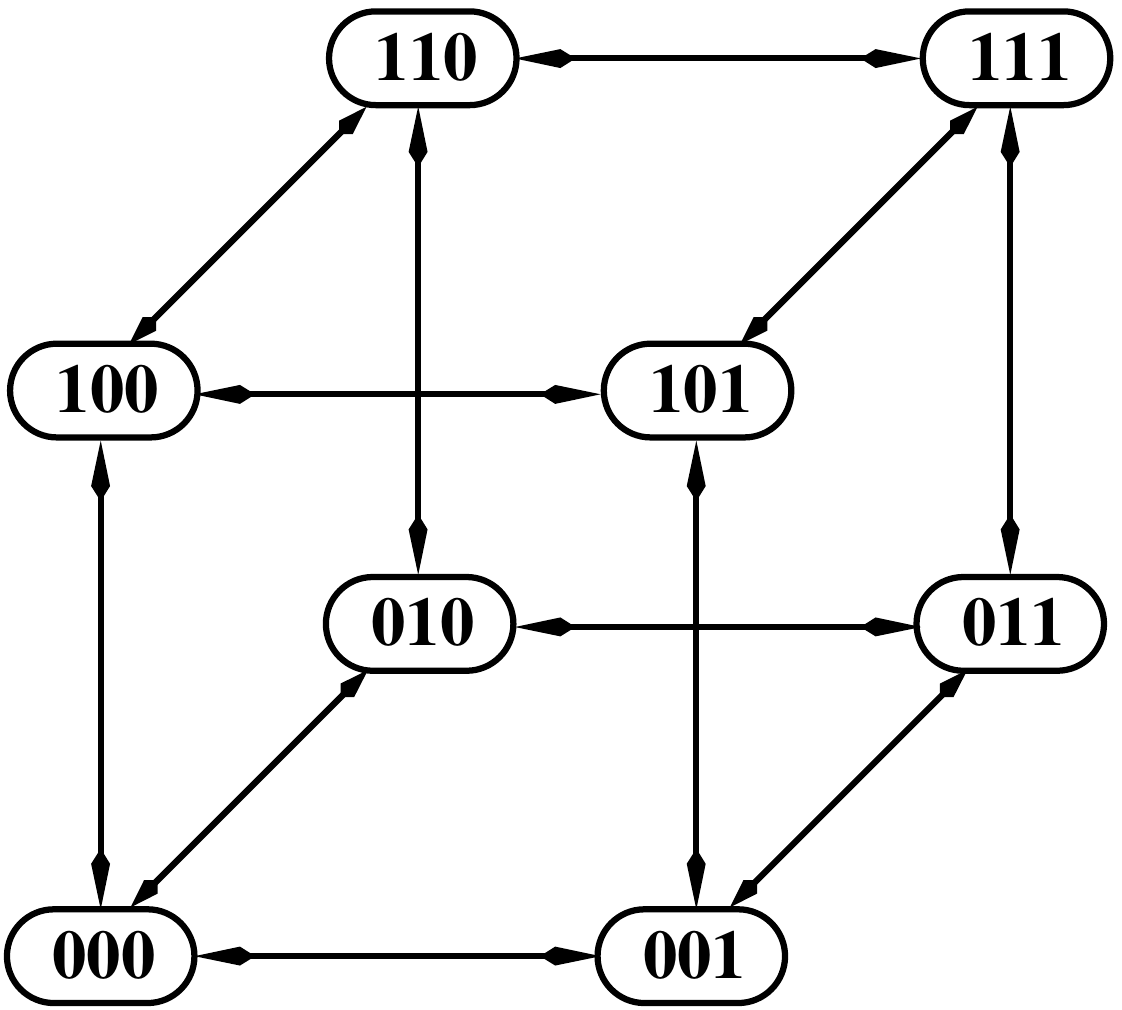}}
\hspace*{-18.5ex}
{\includegraphics[totalheight=0.27\textheight]{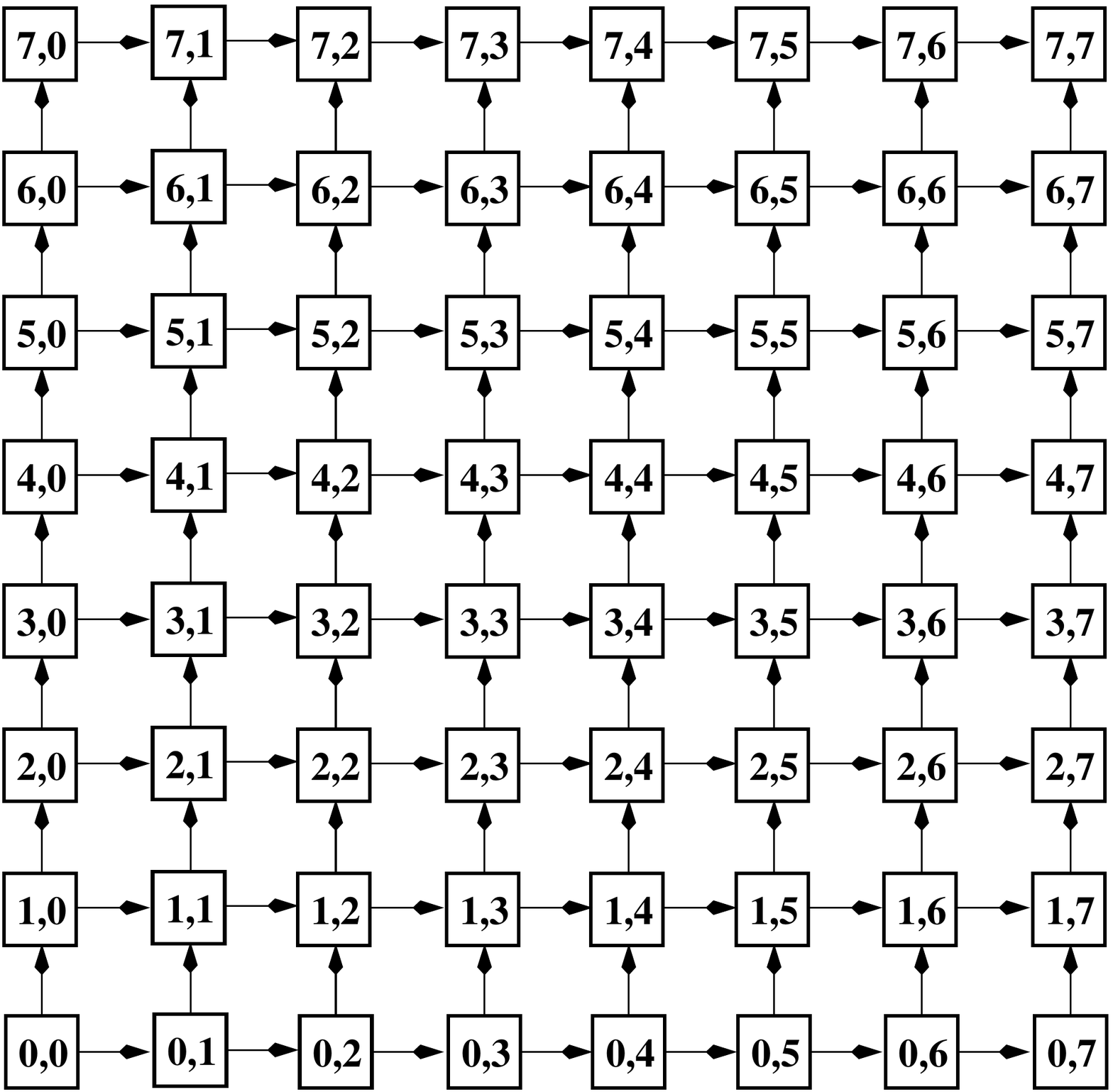}}
\caption{State transition diagrams for $k=3$ bits, and for $n=2$
memory cells with $q=8$ levels.} \label{fig:1 state transition
diagrams} \vspace{-3.5ex}
\end{center}
\end{figure}

An \smash{$(n,k)_q$} flash code $\C$ can now be specified in terms
of two maps: a decoding map $\Delta$~and a~transition map $f$. The
decoding map \smash{$\Delta : \cA_q^n \to \Ftk$} simply indicates
for each cell state vector $\bfx \in V(\cG_n)$~the~value of the
variable vector associated with the corresponding cell state
vector. This map can be, in principle, arbitrary, although it must
be chosen carefully in order to obtain flash codes with good
performance. For each $\al \in \Ftk$, let $\cG_n(\al)$ denote the
set of vertices $\bfx \in V(\cG_n)$ such that $\Delta(\bfx) =
\al$. With this, the transition map \smash{$f : E(\cH_k) {\times}
\cA_q^n \to \cA_q^n \cup \{\E\}$} can be described as follows. For
every edge $(\al,\bet) \in E(\cH_k)$ and every vertex $\bfx \in
\cG_n(\al)$, the value of $f(\al,\bet;\bfx)$ is a vertex $\bfy \in
\cG_n(\bet)$ such that $y_i \,{\ge}\, x_i$ for all $i$ (so that
there exists a path from $\bfx$ to $\bfy$ in $\cG_n$). Or, if no
such vertex exists, $f(\al,\bet;\bfx) = \E$ indicating that a
block erasure is required. It is clear from the definition that
$(\Delta,f)$ map walks in $\cH_k$ onto directed paths in $\cG_n$,
potentially terminating in the block-erasure symbol $\E$.
\begin{defn}
A flash code $\C(\Delta,f)$ \emph{guarantees $t$ writes} if all
walks of length $t$~in~$\cH_k$, starting at the vertex
$(0,0,\dots,0)$, map onto valid paths in $\cG_n$, never producing
the symbol $\E$.
\end{defn}

The weight of a cell state vector \smash{$\bfx \in \cA_q^n$} is
defined to be $w_{\smash{\bfx}}= \sum_{i=1}^nx_i$, and for
convenience will be called the \emph{cell state weight}. We note
that at each write operation the increase in the cell state weight
is at least one, and hence a trivial upper bound for the number of
writes, $t$, is $n(q-1)$.
\begin{defn}
If a flash code guarantees at least $t$ writes before erasing,
then its \emph{write deficiency} (or simply \emph{deficiency}) is
defined to be $\delta= n(q-1)-t$.
\end{defn}

In~\cite{JBB07}, a code for storing two bits is presented. The
code is constructed for arbitrary $n$ and $q$ and guarantees
$t=(n-1)(q-1)+\left\lfloor\frac{q-1}{2}\right\rfloor$ writes. A
general upper bound on $t$ which holds for any $k,\ell,n,q$
($\ell$ is the variable alphabet size, and usually $\ell =2$) is
presented as well, and assures that this two-bits construction is
optimal.
\begin{theorem}\cite{JBB07}
For any code that uses $n$ $q$-level cells and guarantees $t$
writes before erasing, if $n\geq k(l-1)-1$, then $t\leq
(n-k(l-1)+1)\cdot (q-1) + \left\lfloor\frac{ (k(l-1)-1)\cdot
(q-1)}{2}\right\rfloor$; if $n< k(l-1)-1$, then $t\leq
\left\lfloor\frac{ n (q-1)}{2}\right\rfloor$.
\end{theorem}
This bound provides us also with a lower bound on the write
deficiency of flash codes.
\begin{cor}\label{cor:deficiency bound}
For any code that uses $n$ $q$-level cells and guarantees $t$
writes before erasing, the code deficiency satisfies  $\delta\geq
(k(l-1)-1)\cdot (q-1) - \left\lfloor\frac{ (k(l-1)-1)\cdot
(q-1)}{2}\right\rfloor$ if $n\geq k(l-1)-1$, and $\delta\geq
n(q-1) - \left\lfloor\frac{ n (q-1)}{2}\right\rfloor$ if $n<
k(l-1)-1$.
\end{cor}
Furthermore, the bound in~\cite{JBB07} implies that for $n$ large
enough the write deficiency of the code is not dependent on $n$.
\begin{defn}
Let $q$ be a fixed number of cell levels, and $k$ a fixed number
of variables. A family of \smash{$(n_i,k)_q$} flash codes $\C_i$
(where $\lim_{i\rightarrow\infty}n_i = \infty$ ) that guarantees
$t(n_i,q)$ writes is called  \emph{asymptotically optimal} if
$$\lim_{i\rightarrow\infty}\frac{t(n_i,q)}{n_i(q-1)} = 1.$$
\end{defn}

Asymptotically optimal constructions are presented in~\cite{JBB07}
for storing two and three bits. These constructions are later
enhanced and generalized for $3\leq k\leq 6$ bits in~\cite{JB08}.
Also, in~\cite{JB08} a new construction of codes, called
\emph{indexed codes}, is presented and supports the storage of an
arbitrary number of variables. This construction, though shown to
be asymptotically optimal, has deficiency that is dependent on the
number of cells, $n$, and hence is still far from the lower bound
on the deficiency.

The rest of the paper is organized as follows. In
Section~\ref{sec:two bits} we present another optimal construction
for storing two bits. In Section~\ref{sec:basic muldimensional
construction} we demonstrate the basic idea of how to represent an
arbitrary number of bits inside a multidimensional box. Our main
construction is given in Section~\ref{sec:enhanced
multidimensional construction}. We construct codes for efficient
storage of bits such that the write deficiency does not depend on
the number of cells. Finally, conclusions and an open problem are
given in Section \ref{sec:conclusion}.

\section{Another Optimal Construction for Two Bits}\label{sec:two bits}
The construction presented in~\cite{JBB07} for storing two bits
using an arbitrary number of $q$-level cells is optimal. We
present here another optimal construction which we believe is
simpler. In this construction, the leftmost (rightmost) cells of
the memory correspond to the first (second) bit. In writing, if we
change the first (second) bit then we increase by one the leftmost
(rightmost) cell having level less than $q-1$. We repeat this
process until the cells coincide and then the only cell of level
less than $q-1$ represents the two bits. In general, the cell
state vector has the following form:
$(q-1,\ldots,q-1,x_i,0,\ldots,0,x_j,q-1,\ldots,q-1)$, where $0<
x_i,x_j\leq q-1$. We present here the construction for odd values
of $q$, but it is easily modified to handle even values as well.\\
\textbf{Encoding:} We consider the following cases:
\begin{enumerate}
    \item There are at least two cells of level less than $q-1$.
If we change the level of $v_1(v_2)$ then we increase by one the
leftmost (rightmost) cell having level less than $q-1$. If after
this change there is only one cell of level less than $q-1$, then
it has to represent two bits. We change its level such that its
residue modulo $4$ corresponds to the four possible variable
vectors. If the cell level has residue modulo $4$ equal to
$0,1,2,$ or $3$, then the corresponding variable vector is
$(0,0),(0,1),(1,0),$ or $(1,1)$, respectively.

    \item There is only one cell of level less than $q-1$. In this
case the cell represents both bits and we increase its level to
the correct residue modulo $4$ according to the new value of the
variable vector.\\
\end{enumerate}
 \textbf{Decoding:} The equivalent cases are considered:
\begin{enumerate}
    \item There are at least two cells of level less than $q-1$.
Let $i_1(i_2)$ be the index of the leftmost (rightmost) cell
having level less than $q-1$. Then, we decode
$$v_1 = x_{i_1}(\bmod~2),\ \ v_2 = x_{i_2}(\bmod~2).$$

    \item There is one cell of level less than $q-1$, and it is
the $i$-th cell. Then we decode $$v_1 = \left\lfloor
(x_{i}(\bmod~4))/2\right\rfloor, v_2 =
(x_{i}(\bmod~4))(\bmod~2).$$

    \item All cells have level $q-1$. We decode according to the
previous rule with $x_i=q-1$.
\end{enumerate}

For even values of $q$, the construction is very similar. Again,
the last available cell represents two bits. However, now the
value of the two bits, given by the last available cell, is their
relative difference from the value of the two bits given by all
other cells. We note that since $q$ is even, a cell of level $q-1$
represents a bit of value $1$ and not $0$ as we had in the case of
odd $q$. Furthermore, if we use the last available cell up to
level $q-1$ then it will be impossible to distinguish which cell
represents two bits in case all of them are at level $q-1$.
Therefore, we use the last available cell only until level $q-2$.
This construction is optimal as well.
\begin{theorem}
If there are $n$ $q$-level cells, then the code described above
guarantees at least
$t=(n-1)(q-1)+\left\lfloor\frac{q-1}{2}\right\rfloor$ writes
before erasing.
\end{theorem}
\begin{proof}
As long as there is more than one cell of level less than $q-1$,
the cell state weight increases by one after each write. This may
change only after at least $(n-1)(q-1)$ writes. Let us assume that
there is only one available cell of level less than $q-1$ after
$s=(n-1)(q-1)+j$ writes, where $j\geq 0$. Starting at this write,
the different residues modulo $4$ of this cell correspond to the
four possible variable vectors. Therefore, at the $s$-th write, we
also need to increase the level of the last available cell so it
will correspond to the variable vector at the $s$-th write. For
all succeeding writes, if we change the first bit then the cell
level increases by two. If however we change the second bit then
the increase in the cell level alternates between one and three.
Hence, if there are $m$ writes to the last available cell, then
the cell level increases by at most $2m+1$. Consider the case
where the last available cell starts representing the two bits
together. If its level, before it starts representing the two bits
together and after updating its level to correspond to the
variable vector at the $s$-th write, is $x$, then there are at
least $\left\lfloor(q-1-x)/2\right\rfloor$ more writes.
Next, we consider all possible options for the values of
$j$ and the variable vector at the $s$-th write in order to
calculate the number of guaranteed writes before erasing.
\begin{enumerate}
    \item Suppose $j(\bmod~4) = 0$, the value of both bits is $0$,
and the level of the last available cell does not increase at the
$s$-th write. Hence, there are at least
$\left\lfloor(q-1-j)/2\right\rfloor$ more writes and a total of at
least $(n-1)(q-1) + j + \left\lfloor(q-1-j)/2\right\rfloor \geq
(n-1)(q-1) + \left\lfloor(q-1)/2\right\rfloor$ writes.

    \item Suppose $j(\bmod~4) = 1$, one of the bits has value $1$
and the other one $0$. If the variable vector is $(v_1,v_2) =
(0,1)$ then at the $s$-th write the level of the last available
cell does not increase and if it is $(v_1,v_2) = (1,0)$ then its
level increases by one. There are at least
$\left\lfloor(q-1-(j+1))/2\right\rfloor$ more writes, where $j\geq
1$ and a total of at least $(n-1)(q-1) + j +
\left\lfloor(q-2-j)/2\right\rfloor \geq (n-1)(q-1) +
\left\lfloor(q-1)/2\right\rfloor$ writes.

    \item Suppose $j(\bmod~4) = 2$, the value of both bits is $0$
and we increase the level of the last available cell by two at the
$s$-th write. Therefore, there are at least
$\left\lfloor(q-1-(j+2))/2\right\rfloor$ more writes, where $j\geq
2$ and a total of at least $(n-1)(q-1) + j +
\left\lfloor(q-3-j)/2\right\rfloor \geq (n-1)(q-1) +
\left\lfloor(q-1)/2\right\rfloor$ writes.

    \item Suppose $j(\bmod~4) = 3$, one of the bits has value $1$
and the other one $0$. If the variable vector is $(v_1,v_2) =
(0,1)$ then the level of the last available cell increases by two,
and if it is $(v_1,v_2) = (1,0)$ then we increase by three the
level of the last available cell at the $s$-th write. Thus, there
are at least $\left\lfloor(q-1-(j+3))/2\right\rfloor$ more writes,
where $j\geq 3$ and a total of at least $(n-1)(q-1) + j +
\left\lfloor(q-4-j)/2\right\rfloor \geq (n-1)(q-1) +
\left\lfloor(q-1)/2\right\rfloor$ writes.
\end{enumerate}
In any case, the guaranteed number of writes is
$(n-1)(q-1)+\left\lfloor\frac{q-1}{2}\right\rfloor$.
\end{proof}

\section{Basic Multidimensional Construction}\label{sec:basic muldimensional construction}
In this section we start the discussion of how to store an
arbitrary number of bits. We demonstrate a basic construction for
representing the bits inside a multidimensional box. The main
drawback of this construction is its relatively high write
deficiency that depends on the number of cells. In the next
section we will show an alternative construction with a better
deficiency.

Assume we want to store four bits using $n$ $q$-level cells. We
represent the memory as a matrix of $n_1\times n_2$ cells, where
$n_1n_2=n$. In each column two bits are stored. The first and
second bits are stored using the left columns. The leftmost column
is used first, then the second leftmost and so on. Similarly, the
third and fourth bits are stored using the right columns
right-to-left. In each column we store the bits from the opposite
directions as in the previous two-bits construction. However, in
this case we don't use the last available cell to represent two
bits, but leave it as a separation cell. Assume we change the
value of one of the first two bits, if it is possible to update
this change in the current column that represents these bits, we
do so. Otherwise, and if there is at least one more column for
separation, we use the next column. An example of the memory state
of this construction is demonstrated in Figure~\ref{fig:basic
construction example}. The worst case scenario for the number of
writes before erasing occurs when:
\begin{enumerate}
    \item One column is used for separation.

    \item Another column is only partially used and represents one
write operation. That is, there was only one write to this column
and still it is impossible to update the current write in this
column.

    \item In two other previous columns there is one more cell
that is also only partially used and represents one write
operation.
\end{enumerate}
Hence, the guaranteed number of writes is $(n_1-1)(n_2-1)(q-1) -
\left((n_1-1)(q-1)-1\right) - 2\left((q-1)-1\right)$, where
$n_1n_2=n$. Another example of the memory state that corresponds
to the worst case scenario is given in Figure~\ref{fig:basic
construction worst case}.
\begin{figure}
\begin{center}
\subfigure[]{ \setlength{\unitlength}{2.5cm}
\includegraphics[totalheight=0.26\textheight]{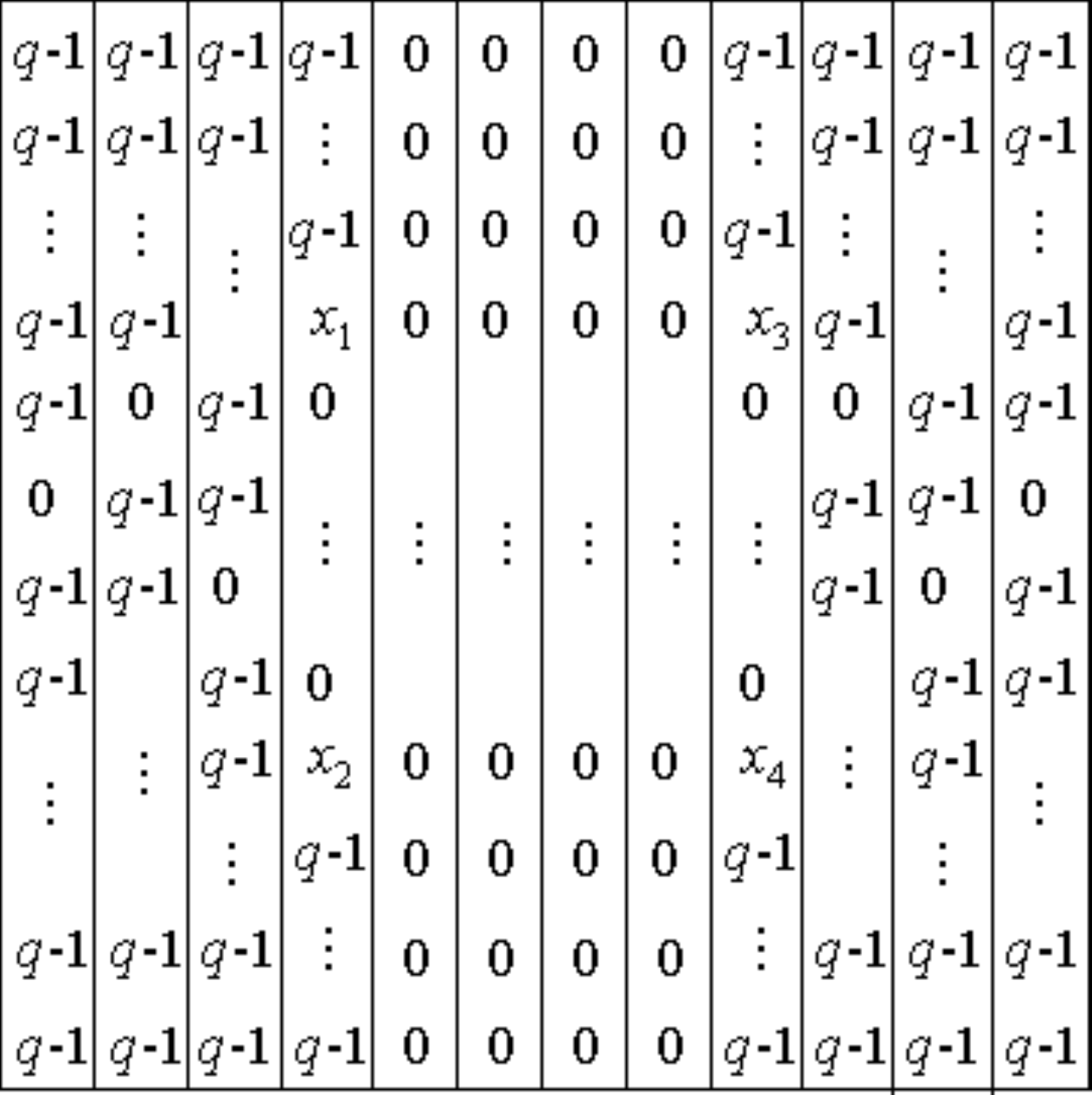}
\label{fig:basic construction example}}
\subfigure[]{\setlength{\unitlength}{2.5cm}
\includegraphics[totalheight=0.26\textheight]{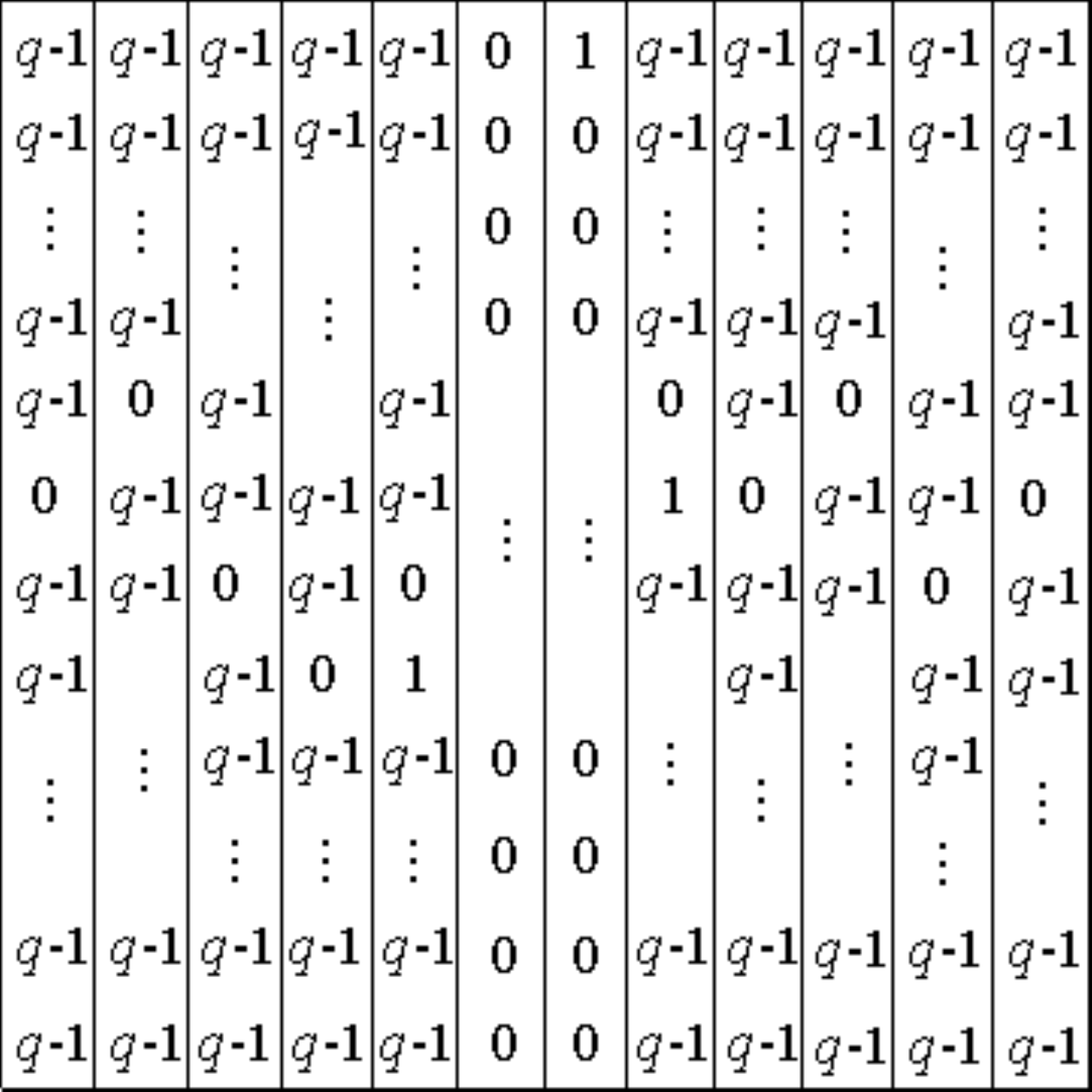}
\label{fig:basic construction worst case}}
\caption{Figure~\ref{fig:basic construction example} describes an
example of the memory state for the basic multidimensional
construction, and Figure \ref{fig:basic construction worst case}
demonstrates an example for the worst case scenario of the
guaranteed number of writes when it is impossible to write the
first bit.} \vspace{-0.3in}
\end{center}
\end{figure}

The generalization of this construction to three dimensions
supports the storage of up to eight bits. Each plane stores four
bits. The lower planes represent the first four bits and the upper
planes represent the last four bits. In each plane we use the
previous construction in order to represent four bits. We can use
all the columns in each plane except for one that is left for
separation between the two groups of two bits in this plane. Also,
one more plane is used for separation between the two groups of
four bits. The equivalent worst case scenario for the number of
writes before erasing occurs as follows:
\begin{enumerate}
    \item One plane is used for separation.

    \item Another plane is partially used and represents only one
write operation.

    \item In two previous planes there is one more column that
represents only one write and two more cells that represent only
one write as well.
\end{enumerate}
Therefore, the guaranteed number of writes is
$(n_1-1)(n_2-1)(n_3-1)(q-1) - \left((n_1-1)(n_2-1)(q-1)-1\right) -
2\left((n_1-1)(q-1)-1\right)-4((q-1)-1)$, where $n_1n_2n_3=n$.

In general, using a $D$-dimensional box we can store $2^D$ bits,
and we can show that the number of guaranteed writes is
\begin{align*}
& \prod_{i=1}^{D}(n_i-1)(q-1)-\sum_{i=1}^{D-1}2^{i-1}\left(\prod_{j=1}^{D-i}(n_j-1)(q-1)-1\right)\\
& - 2^{D-1}\left((q-1)-1\right),
\end{align*}
where $n_1n_2\cdots n_D=n$. It is also possible to show that this
code is asymptotically optimal. However, its deficiency depends on
$n$ and therefore is not close enough to the lower bound on the
deficiency. Next, we will show how to modify this construction in
order to obtain a deficiency that is only dependent on the number
of bits $k$ and the number of cell levels $q$.

\section{Enhanced Multidimensional Construction}\label{sec:enhanced multidimensional construction}
The last construction of flash codes demonstrates the idea of how
to use a multidimensional box in order to represent multiple bits.
Even though its asymptotic behavior is optimal, there is still a
large gap between its write deficiency and the lower bound on the
deficiency. The high deficiency mainly results from the separation
cell in each column, the separation column in each plane, and in
general from the separation hyperplane in each dimension.

The following construction of flash codes shows how to improve the
deficiency. In order to reduce the deficiency from the extra
separation hyperplane in the last dimension, the length of each
dimension, besides the last one, should be as small as possible,
for example we want to choose $n_i = 2,3$ for $1\leq i\leq D-1$.
We also want to store the bits differently so that in each
dimension it is possible to take advantage of all cell levels
before using the next dimension. Another advantage of using small
dimension lengths is that the rate of the code (defined as
$k/(n\log_2\!q)$\,) is enhanced as well. We show a construction
where each dimension length, other than the last one, is two. If
we want to store $k=2^D$ binary bits then we show how to store
$2^{D-1}$ of them inside $(D-1)$-dimensional boxes of size
$n_1\times n_2 \times \cdots \times n_{D-1} = 2\times 2\times
\cdots \times 2$. Then, we use a $D$-dimensional box of size
$n_1\times n_2 \times \cdots \times n_D= 2\times 2\times \cdots
\times 2 \times n_{D}$, where $n_{D} \geq 3$, which consists of
$n_{D}$ $(D-1)$-dimensional boxes. For convenience, we call every
multidimensional box, of any dimension, whose edges are of length
$2$, a \emph{block}.

The construction is recursive. We first present how to store two
bits using two-cell blocks. Then, we use this construction in
order to store four bits using two-dimensional four-cell blocks of
size $2\times 2$. Using the four-bits construction, it is possible
to store eight bits in a three-dimensional eight-cell blocks of
size $2\times 2\times 2$. In general, the construction for storing
$2^{i-1}$ bits in $(i-1)$-dimensional blocks, where $i\geq 3$, is
utilized in order to store $2^{i}$ bits in $i$-dimensional blocks
of $2^{i}$ cells each. We show in detail the basic constructions
for storing two and four bits as these constructions are the
building blocks for the arbitrary recursive construction. An
analysis of the construction deficiency is given as well and it is
shown that the write deficiency order is $O(k^2q)$.

Like the two-bits construction, this construction is presented for
odd values of $q$, and it is possible to modify it in order to
support even values. However, the deficiency is larger when $q$ is
even.
\subsection{Two-Bits Construction}
Our point of departure for these codes is a basic construction for storing two bits using blocks of two cells.\\
\textbf{Encoding:}
\begin{enumerate}
    \item As long as the number of writes in the block is no
greater than $q-1$:
    \begin{enumerate}
        \item If the first bit is changed then the left cell is
raised by one.

        \item If the second bit is changed then the right cell is
raised by one.
    \end{enumerate}

    \item Starting at the $q$-th write, it may happen that for
some cell state vectors of the block, it is possible to write only
one of the bits. In this case, if the other bit is changed then a
new block is used.
    \begin{enumerate}
        \item If the cell state vector of the block is of the form
$(q-1,x)$, where $x< q-1$, then only the first bit can be written
to this block, and the level of the second cell is raised by one.

        \item If the cell state vector of the block is of the form
$(x,q-1)$, where $x< q-1$, then only the second bit can be written
at the next write, and the level of the first cell is raised by
one.

        \item If the cell state vector of the block is of the form
        $(x_1,x_2)$, where $x_1<q-1, x_2<q-1$, then both bits can
        be written at the next step. If the first (second) bit is
        changed then the second (first) cell is raised by one.
    \end{enumerate}
\end{enumerate}
\textbf{Decoding:}
\begin{enumerate}
    \item If the cell state vector is of the form $(x_1,x_2)$,
where $0\leq x_1,x_2\leq q-1$, and $x_1+x_2\leq q-1$ then the
variable vector is $$(v_1,v_2) = (x_1(\bmod~2),x_2(\bmod~2)).$$

    \item If the cell state vector is of the form $(x_1,x_2)$,
where $0\leq x_1,x_2\leq q-1$, and $x_1+x_2> q-1$ then the
variable vector is $$(v_1,v_2) = (x_2(\bmod~2),x_1(\bmod~2)).$$
\end{enumerate}
An example of this construction for $q=5$ is given in
Figure~\ref{fig:2 vars 2 cells}.
\begin{figure}
\begin{center}
\setlength{\unitlength}{2.5cm}
\includegraphics[totalheight=0.27\textheight]{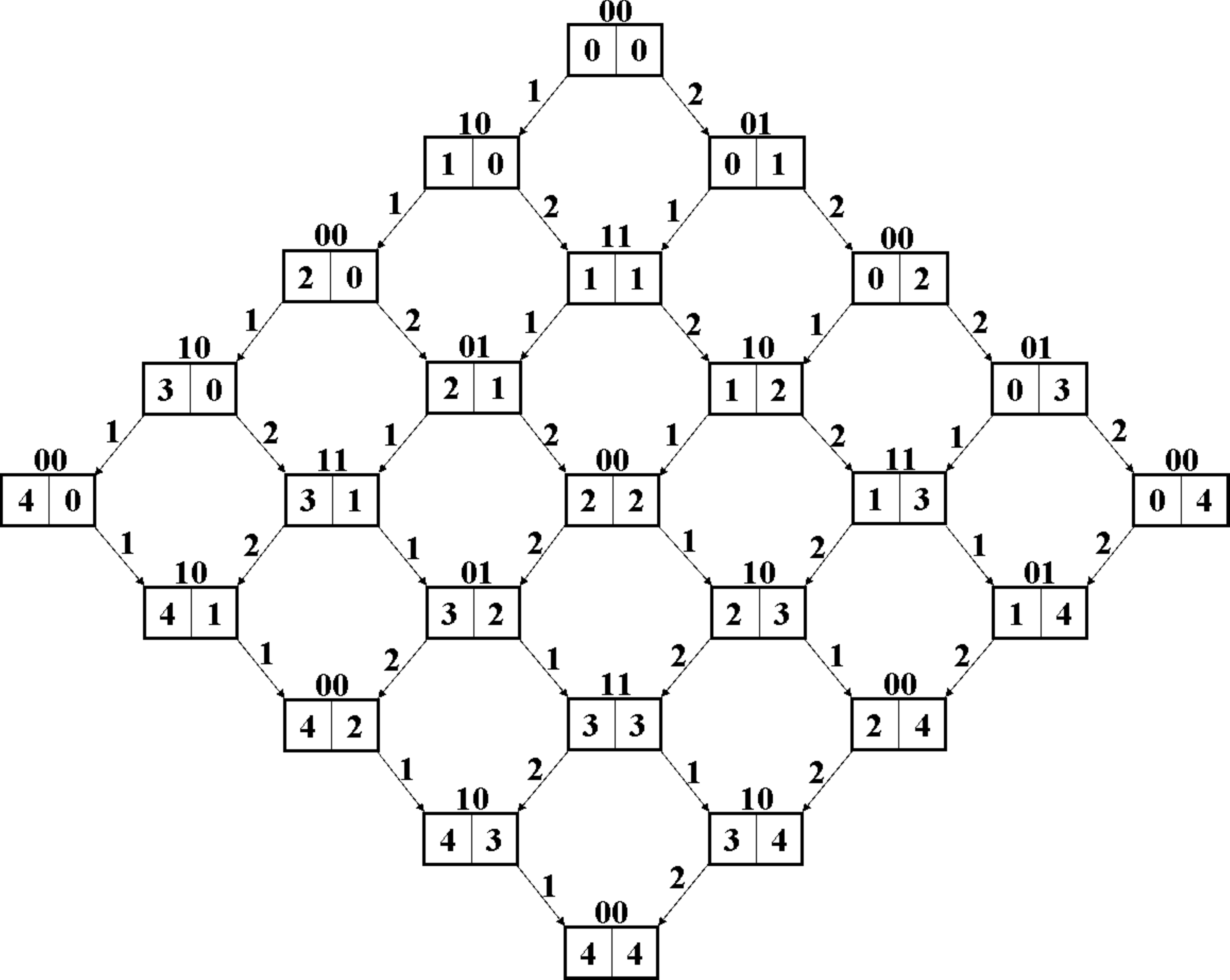}
\caption{State diagram for flash codes storing two bits in a block
of two $\textrm{5}$-level cells. The numbers in each block, above
the block, and next to each edge represent the cell state vector,
variable vector, and the written bit, respectively.} \label{fig:2
vars 2 cells}
\end{center}
\end{figure}
We note that as long as the number of writes is no greater than
$q-1$ then it is possible to write both bits. Only at the $q$-th
write might it happen that writing will continue in the next
block.

By abuse of terminology we use the following definitions for a
block of any size:
\begin{enumerate}
    \item A block is called \emph{empty} if all its cells are at
level zero.

    \item A block is called \emph{full} if all its cells are at
level $q-1$.

    \item A block is called \emph{active} if it is neither empty
nor full.
\end{enumerate}
\begin{lemma}\label{lem:2 cells blocks loss}
For the two-bits construction, at any write operation, there are
at most two active blocks and at most $A_1 = (q-1)+2(q-1)-1 =
3(q-1)-1$ levels that are not used in these two blocks.
\end{lemma}
\begin{proof}
If a new block is used then the previous block uses at least $q-1$
levels, and there are no more active blocks. Therefore, at most
$q-1$ levels are not used in the previous block and $2(q-1)-1$ in
the new block.
\end{proof}

\subsection{Four-Bits Construction}
Next, the four-bits case is considered. We use blocks of four
cells. Each block is divided into two sub-blocks of two cells
each, and each such sub-block is a column that stores two bits
according to the previous construction. The block can either store
the first and second bits together or the third and fourth bits
together, i.e., it is impossible to store all four bits together
in the same block. If the block stores the first and second bits
then the sub-blocks are written left-to-right, and if the block
stores the third and fourth bits then the sub-blocks are written
right-to-left.
$$ 1,2\rightarrow
\begin{tabular}{|c||c|}
  \hline
  * & * \\
  * & * \\
  \hline
\end{tabular}\leftarrow 3,4.$$
In this construction we have another distinction between storing
the first and second bits and the third and fourth bits. Each
sub-block represents two bits according to the previous
construction, but the order the bits are written in the sub-blocks
is changed. If the block represents the first and second bits then
both sub-blocks represent the bits in the same way:
$$\begin{tabular}{|c||c|}
  \hline
  1 & 1 \\
  2 & 2 \\
  \hline
\end{tabular}$$
However, if the block represents the third and fourth bits then
the representation order of the bits in the two sub-blocks is
changed as follows:
$$ \begin{tabular}{|c||c|}
  \hline
  4 & 3 \\
  3 & 4 \\
  \hline
\end{tabular}$$
\textbf{Encoding:}
\begin{enumerate}
    \item The block can either represent the bits $1,2$ or the
bits $3,4$. In each sub-block (column) we use the previous
construction in order to represent two bits.

    \item If the block represents the bits $1,2$, then we write
the two sub-blocks left-to-right, and the bits are stored
similarly in the two sub-blocks.

    \item If the block represents the bits $3,4$, then we write
    the two sub-blocks right-to-left. For the right sub-block we
    represent the two bits where the third (fourth) bit is
    considered to be the first (second) bit in the two bits
    construction, $$3\leftrightarrow 1, 4\leftrightarrow 2.$$
    However, for the left sub-block we change the order of the
    bits, i.e., the third (fourth) bit is considered to be the
    second (first) bit in the two-bits construction,
    $$3\leftrightarrow 2, 4\leftrightarrow 1.$$

    \item If the block represents the bits $1,2$ ($3,4$) then the
right (left) sub-block cannot be full before the left (right)
sub-block is full.
\end{enumerate}

\textbf{Decoding:}
\begin{enumerate}
    \item For every active block we first determine, according to
the encoding rules, whether it represents the bits $1,2$ or $3,4$:
    \begin{enumerate}
        \item If the right (left) sub-block is empty then the
block represents the bits $1,2$ ($3,4$).

        \item If the left (right) sub-block is full then the block
represents the bits $1,2$ ($3,4$).

        \item If both sub-blocks are active then according to the
decoding procedure of the two-bits construction we can decide for
each sub-block whether it is in a state that enables to write both
bits, only the first bit or only the second bit.
            \begin{enumerate}
                \item If the right (left) sub-block is in a state
that enables to write both bits then the block represents the bits
$1,2$ ($3,4$).

                \item Assume the states of both sub-blocks
enable to write only one bit. If both bits from the two sub-blocks
are first bits or both bits are second bits then the block
represents the bits $3,4$. Otherwise, it represents the bits
$1,2$.
            \end{enumerate}
    \end{enumerate}

    \item If the block represents the bits $1,2$ or $3,4$, we can
decode their values from the two sub-blocks using the decoding
procedure of the two-bits construction.

    \item The value of each bit is the XOR of its values from all
blocks.
\end{enumerate}
An example of this construction for $q=3$ is given in
Figure~\ref{fig:4 vars 4 cells}.
\begin{figure}
\begin{center}
\setlength{\unitlength}{2.5cm}
\includegraphics[totalheight=0.26\textheight]{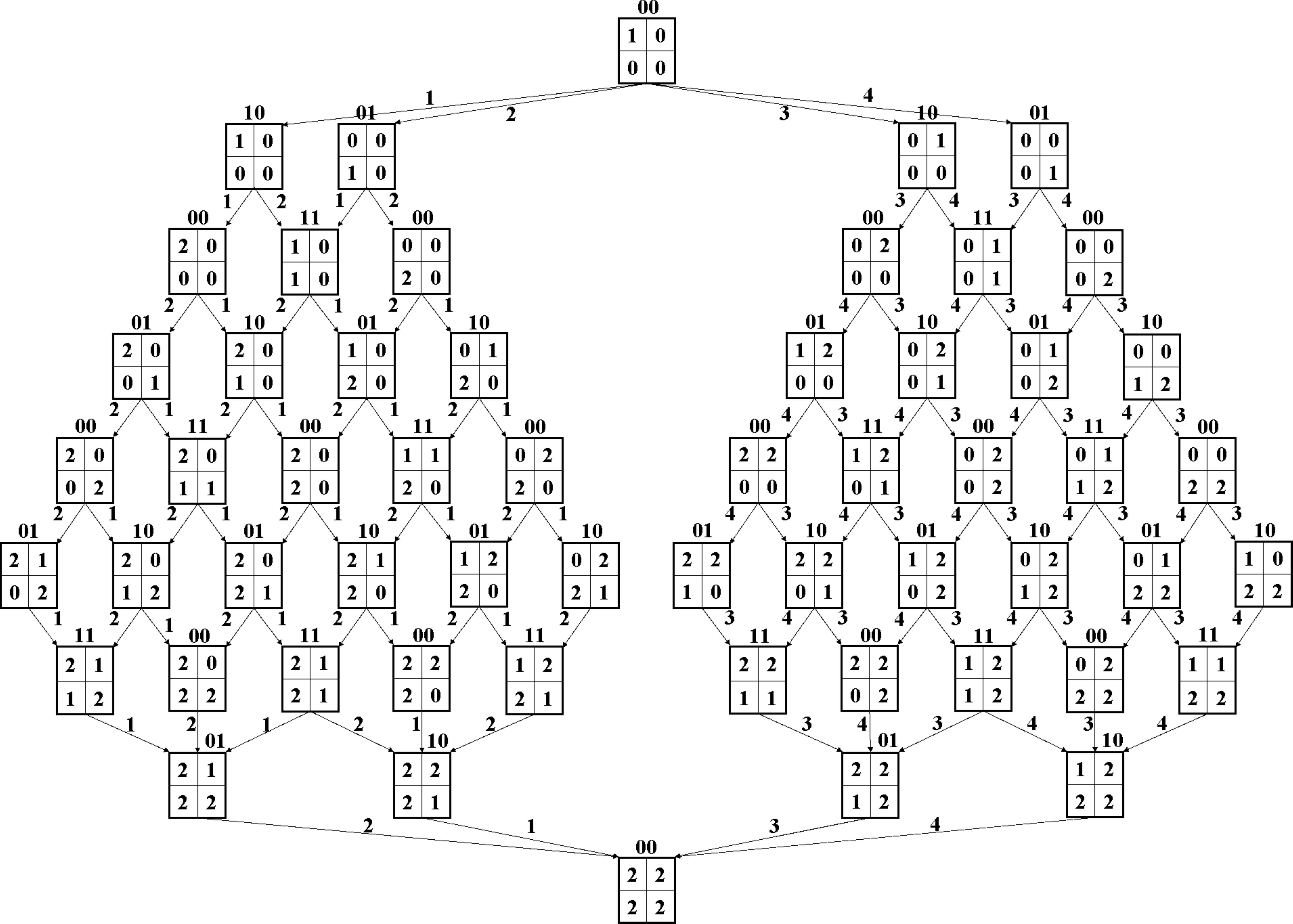}
\caption{State diagram for flash codes storing four bits in a
four-cell block. The numbers in each block, above the block, and
next to each edge represent the cell state vector, variable
vector, and the written bit, respectively} \label{fig:4 vars 4
cells}
\end{center}
\vspace{-0.1in}
\end{figure}
\begin{lemma}\label{lem:4 cells blocks loss}
For the four-bits construction, after any write operation, there
are at most six active blocks and at most $A_2 =
2\left((q-1)+1+4(q-1)-1\right) = 10(q-1)$ levels that are not used
in these four blocks.
\end{lemma}
\begin{proof}
For each pair of bits there are at most three active blocks, a new
block and two previous active blocks that could not be full before
starting the new block. In the two previous blocks at most
$(q-1)+1$ levels are not used, corresponding to the case that the
first sub-block does not use one level and the second block does
not use $(q-1)$ levels. For example, if the blocks represent the
first and second bits, the two previous active blocks can be of
the following form, that enables to write only the second bit:
$$ \begin{tabular}{|c|c|}
  \hline
  $q-1$ & $q-2$ \\
  $q-1$ & $q-1$ \\
  \hline
\end{tabular}\ \
\begin{tabular}{|c|c|}
  \hline
  $q-1$ & $0$ \\
  $q-1$ & $q-1$ \\
  \hline
\end{tabular}$$ In the new block at most $4(q-1)-1$ levels are not used.
Therefore, there are at most six active blocks and at most $A_2 =
2\cdot((q-1)+1 + 4(q-1)-1) = 10(q-1)$ levels that are not used in
these blocks.
\end{proof}

\subsection{Construction for Arbitrary Number of Bits}
We are now ready to present the general construction of flash
codes storing an arbitrary number of bits. First, we briefly
describe how to represent eight bits and then give the general
construction.

In order to store eight bits we use a block of eight cells which
is a three-dimensional box of size $2\times 2 \times 2$. In fact,
the block consists of two sub-blocks of four cells each that can
be considered as two concatenated sub-blocks of size $2\times 2$.
$$\begin{tabular}{|c|c||c|c|}
  \hline
  * & * & * & * \\
  * & * & * & * \\
  \hline
\end{tabular} $$
Each block can either represent the bits $1,\ldots,4$ or
$5,\ldots,8$ according to the following order:
$$ 1,\ldots,4\rightarrow \begin{tabular}{|c|c||c|c|}
    \hline
    \scriptsize {$\underrightarrow{1,2}$} & \scriptsize {$\underleftarrow{3,4}$} & \scriptsize {$\underrightarrow{5,6}$} & \scriptsize {$\underleftarrow{7,8}$} \\
    * & * & * & * \\
    \hline
    \end{tabular} \leftarrow 5,\ldots,8. $$
The bits $1,\ldots,4$ ($5,\ldots,8$) write the two $2\times 2$
sub-blocks left-to-right (right-to-left). Each sub-block
represents four bits according to the four-bits construction. More
rules are used to decide whether the code represents the bits
$1,\ldots,4$ or $5,\ldots,8$, and are described in detail below
for the arbitrary case. 

For representing $2^i$ bits, we assume that there is a
construction for storing $2^{i-1}$ bits in blocks of $2^{i-1}$
cells. We use blocks of $2^i$ cells that consist of two sub-blocks
of $2^{i-1}$ cells. We assume that for the $2^{i-1}$-bits
construction there are at most $3\cdot 2^{i-2}$ active blocks and
at most $A_{i-1}$ levels that are not used in these blocks. \\
\textbf{Encoding:}
\begin{enumerate}
    \item The block can either represent the bits
$1,\ldots,2^{i-1}$ or the bits $(2^{i-1}+1) ,\ldots,2^{i}$
according to the following order:
    \renewcommand{\tabcolsep}{0.09cm}
    $$\hspace{-0.27in} \tiny{\begin{array}{c} 1\ldots \\ 2^{i-1}\end{array}\rightarrow} \begin{tabular}{|c|c||c|c|}
    \hline
    \tiny {$1\ldots$} & \tiny {$(2^{i-2}+1)$} & \tiny {$(2^{i-1}+1)\ldots$} & \tiny {$(2^{i-1}+2^{i-2}$} \\
    \tiny {$\underrightarrow{2^{i-2}}$} & \tiny {$\underleftarrow{\ldots2^{i-1}}$} & \tiny {$\underrightarrow{(2^{i-1}+2^{i-2})}$} & \tiny{$\underleftarrow{+1) \ldots2^{i}}$}  \\
    \hline \end{tabular}\leftarrow \begin{array}{c} (2^{i-1}+1) \\
\cdots2^{i}\end{array}$$
    In each sub-block, $2^{i-1}$ bits are represented according to
the recursive construction.

    \item Assume the block represents the bits $1,\ldots,2^{i-1}$,
    \begin{enumerate}
        \item The sub-blocks are written left-to-right.

        \item It is possible to use the right sub-block only if
        the left sub-block is full.
    \end{enumerate}

     \item Assume the block represents the bits $(2^{i-1}+1),\ldots, 2^{i}$,
    \begin{enumerate}
        \item The sub-blocks are written right-to-left.

        \item It is possible to use the left sub-block only if the
        right sub-block is full.

        \item The $2^{i-1}$ bits $(2^{i-1}+1),\ldots, 2^{i}$ are
        stored as if they were the bits $1,\ldots,2^{i-1}$,
        ~where~the~$(2^{i-1}+j)$~-th bit, $1\leq j \leq 2^{i-1}$
        is considered to be the $j$-th bit.
    \end{enumerate}
 \end{enumerate}
\textbf{Decoding:}
\begin{enumerate}
    \item For every active block we first determine, according to the
encoding rules, which group of $2^{i-1}$ bits it represents. If
the left (right) sub-block is full or active, then the block
represents the bits $1,\ldots, 2^{i-1}$
($(2^{i-1}+1),\ldots,2^{i}$).

    \item If the block represents the bits $1,\ldots, 2^{i-1}$ or
$(2^{i-1}+1),\ldots, 2^{i}$, we can decode their value from the
two sub-blocks using the decoding procedure of the $2^{i-1}$-bits
construction.

    \item The value of each bit is the XOR of its values from all
blocks.
\end{enumerate}
\begin{lemma}\label{lem:2^i cells blocks loss}
For flash codes storing $2^i$ bits, after any write operation,
there are at most $3\cdot 2^{i-1}$ active blocks and at most
$$A_i=2A_{i-1} +3/4\cdot (q-1) 4^i$$ levels that are not used in these blocks.
\end{lemma}
\begin{proof}
Each of the blocks that represents the bits $1,\ldots,2^{i-1}$ can
be considered as a pair of sub-blocks containing $2^{i-1}$ cells,
such that each sub-block represents the bits $1,\ldots,2^{i-1}$.
According to the recursive construction at most $3\cdot 2^{i-2}$
sub-blocks are active and at most $A_{i-1}$ levels are not used in
these sub-blocks. If all these sub-blocks happen to be the left
ones in their containing blocks then there are $3\cdot 2^{i-2}$
more sub-blocks that are empty, which are the corresponding right
sub-blocks in each block. In each such a sub-block, at most
$2^{i-1}(q-1)$ levels are not used. Hence, in these sub-blocks at
most $$B_i= 3\cdot 2^{i-2}\cdot (q-1)2^{i-1} = 3/8\cdot (q-1)4^i$$
levels are not used. The same analysis is applied to the blocks
that represent the bits $(2^{i-1}+1),\ldots,2^{i}$. Therefore, for
all the bits, there are at most $3\cdot 2^{i-1}$ active blocks and
at most
$$A_i = 2A_{i-1} + 2B_i = 2A_{i-1} +3/4\cdot (q-1)4^i$$ levels that are not used in these blocks.

\end{proof}
\subsection{Deficiency Analysis}
\begin{lemma}
For $i\geq 2$ we have
$$A_i = 3/2\cdot (q-1)4^{i}-7/2\cdot (q-1)2^{i}.$$
\end{lemma}
\begin{proof}
We prove the correctness of the expression for $A_i$ by induction.
According to Lemma~\ref{lem:4 cells blocks loss}, we have $A_2 =
10(q-1)$ which is also given by this expression. Assume
$A_{i-1}=3/2\cdot (q-1)4^{i-1}-7/2\cdot (q-1)2^{i-1}$, for $i\geq
3$, then according to Lemma~\ref{lem:2^i cells blocks loss}
\begin{align*}
& A_i = 2A_{i-1} + 3/4\cdot 4^i(q-1) \\
& = 2  \left( 3/2\cdot (q-1)4^{i-1}-7/2\cdot (q-1)2^{i-1}  \right) + 3/4\cdot 4^i(q-1) \\
& =3/2\cdot (q-1)4^{i}-7/2\cdot (q-1)2^{i}.
\end{align*}
\end{proof}
\begin{theorem}
If the flash codes represent $k=2^D$ bits then the deficiency
$\delta_D$ satisfies
$$\delta_D  = 2A_{D-1} +1 = 3/4\cdot (q-1)k^2-7/2\cdot (q-1)k+1.$$
\end{theorem}
\begin{proof}
In order to represent $k=2^D$ bits a $D$-dimensional box of size
$2\times 2\times \cdots \times 2\times n_D$ is used. The
multidimensional box is considered~as~an~array~of~$n_D$
$(D-1)$-dimensional boxes, called blocks. The first (last)
$2^{D-1}$ bits are represented using the blocks left-to-right
(right-to-left), and there is one block for separation. In each
$(D-1)$-dimensional box we use the construction to represent
$2^{D-1}$ bits. Writing stops if we need to start using a new
block but it is the last separation block. According to
Lemma~\ref{lem:2^i cells blocks loss} at most $A_{D-1}$ levels are
not used in the active blocks for each group of bits. Also, it is
impossible to use the last separation block, and hence at most
$2A_{D-1}+1$ levels are not used in the worst case, where
\begin{align*}
& 2A_{D-1}+1 \\
& = 2\left(3/2\cdot (q-1)4^{D-1}-7/2\cdot (q-1)2^{D-1}\right)+1\\
& = 3/4\cdot (q-1)4^{D}-7/2\cdot (q-1)2^{D} +1 \\
& = 3/4\cdot (q-1)k^2-7/2\cdot (q-1)k+1.
\end{align*}
\end{proof}

For even values of $q$, we consider every two cells of level $q$
as one cell of level $q'=2q-1$, and we can apply the construction
for odd values of $q$. The code deficiency becomes $3/2\cdot
(q-1)k^2-7\cdot (q-1)k+1$.

\section{Conclusion}\label{sec:conclusion}
In~\cite{JBB07}, the problem of coding to minimize block erasures
in flash memories was first presented. In this work we show an
optimal construction of flash codes for storing two bits. We
believe that our construction is simpler than an earlier optimal
construction presented in~\cite{JBB07}. Our main contribution is
an efficient construction of codes that support the storage of any
number of bits. We show that the order of the code deficiency is
$O(k^2q)$, which is an improvement upon the equivalent
construction in~\cite{JB08}. The upper bound in~\cite{JBB07} on
the guaranteed number of writes implies that the order of the
lower bound on the deficiency is $O(kq)$. Therefore, there is a
gap, which we believe can be reduced, between the write deficiency
orders of our construction and the lower bound.

\section*{Acknowledgment}
The authors wish to thank Hilary Finucane and Michael Mitzenmacher
for pointing out errors in an earlier version of the paper .


\begin{thebibliography}{10}
\providecommand{\url}[1]{#1} \csname url@rmstyle\endcsname
\providecommand{\newblock}{\relax}
\providecommand{\bibinfo}[2]{#2}
\providecommand\BIBentrySTDinterwordspacing{\spaceskip=0pt\relax}
\providecommand\BIBentryALTinterwordstretchfactor{4}
\providecommand\BIBentryALTinterwordspacing{\spaceskip=\fontdimen2\font
plus \BIBentryALTinterwordstretchfactor\fontdimen3\font minus
\fontdimen4\font\relax}
\providecommand\BIBforeignlanguage[2]{{%
\expandafter\ifx\csname l@#1\endcsname\relax
\typeout{** WARNING: IEEEtran.bst: No hyphenation pattern has been}%
\typeout{** loaded for the language `#1'. Using the pattern for}%
\typeout{** the default language instead.}%
\else \language=\csname l@#1\endcsname \fi #2}}

\bibitem{BJB07}
{V.\,Bohossian, A.\,Jiang, and J.\,Bruck}, ``Buffer coding for
asymmetric multi-level memory,'' in {\em Proceedings\ IEEE
International Symposium on Information Theory}, Nice, France, June
2007.

\bibitem{CGOZ99}
{P.\,Cappelletti, C.\,Golla, P.\,Olivo, and E.\,Zanoni} (Editors),
{\em Flash memories},~Boston: Kluwer Academic, 1999.

\bibitem{CSBB07}
{Y.\,Cassuto, M.\,Schwartz, V.\,Bohossian, and J.\,Bruck}, ``Codes
for multi-level flash memories: correcting asymmetric
limited-magnitude errors,'' in {\em Proceedings\ IEEE
International Symposium on Information Theory}, Nice, France, June
2007.

\bibitem{CGM86}
{G.D.\,Cohen, P.\,Godlewski, and F.\,Merkx}, ``Linear binary code
for write-once memories,'' {\em IEEE Trans.\ Inform.\ Theory,}
vol.\,32, pp.\,697-700, October 1986.

\bibitem{ER99}
{B.\,Eitan and A.\,Roy}, ``Binary and multilevel flash cells,'' in
{Flash Memories, P. Cappelletti, C. Golla, P. Olivo, E. Zanoni
Eds. Kluwer}, pp. 91-152, 1999.

\bibitem{FS84}
{A.\,Fiat and A.\,Shamir}, ``Generalized write-once memories,''
{\em IEEE Trans.\ Inform.\ Theory,} vol.\,30, pp.\,470--480,
September 1984.

\bibitem{GT05}
{E.\,Gal and S.\,Toledo}, ``Algorithms and data structures for
flash memories,'' {\em ACM Computing Surveys}, vol.\,37, no.\,2,
pp.\,138--163, June 2005.

\bibitem{GCKT03}
{S.\,Gregori, A.\,Cabrini, O.\,Khouri, and G.\,Torelli}, ``On-chip
error correcting techniques for new-generation flash memories,''
in {\em Proceedings of The IEEE}, vol. 91, no. 4, pp. 602-616,
April 2003.

\bibitem{GLR03}
{M.\,Grossi, M.\,Lanzoni, and B.\,Ricco}, ``Program schemes for
multilevel flash memories,'' in {\em Proceedings of the IEEE},
vol. 91, no. 4, pp. 594-601, April 2003.

\bibitem{Jiang-Allerton}
{A.\,Jiang}, ``Information storage in flash memories with floating
codes,'' in {\em Proceedings\ 45-th Annual Allerton Conference on
Communication, Control and Computing}, Monticello, IL, September
2007.

\bibitem{J07}
{A.\,Jiang}, ``On the generalization of error-correcting WOM
codes,'' in {\em Proceedings\ IEEE International Symposium on
Information Theory}, Nice, France, June 2007.

\bibitem{JBB07}
{A.\,Jiang, V.\,Bohossian, and J.\,Bruck}, ``Floating codes for
joint information storage in write asymmetric memories,'' in {\em
Proceedings\ IEEE International Symposium on Information Theory},
Nice, France, June 2007.

\bibitem{JB08}
{A.\,Jiang and J.\,Bruck}, ``Joint coding for flash memory
storage,'' in {\em Proceedings\ IEEE International Symposium on
Information Theory}, Toronto, Canada, July 2008.

\bibitem{JMSB08}
{A.\,Jiang, R.\,Mateescu, M.\,Schwartz, and J.\,Bruck}, ``Rank
modulation for flash memories,'' in {\em Proceedings\ IEEE
International Symposium on Information Theory}, Toronto, Canada,
July 2008.

\bibitem{JSB08}
{A.\,Jiang, M.\,Schwartz, and J.\,Bruck}, ``Error-correcting codes
for rank modulation,'' in {\em Proceedings\ IEEE International
Symposium on Information Theory}, Toronto, Canada, July 2008.

\bibitem{KS67}
{D.\,Kahng and S.M.\,Sze}, ``A floating-gate and its application
to memory devices,'' {\em Bell Systems Tech.\ J.}, vol.\,46,
no.\,4, pp.\,1288-1295, 1967.

\bibitem{MLF08}
{M.\,Mitzenmacher, Z.\, Liu, and H.\, Finucane}, ``Designing
floating codes for expected performance,'' in {\em Proceedings\
46-th Annual Allerton Conference on Communication, Control and
Computing}, Monticello, IL, September 2008.

\bibitem{RS82}
{R.L.\,Rivest and A.\,Shamir}, ``How to reuse a write-once
memory,'' {\em Information and Control}, vol.\,55, nos.\,1--3,
pp.\,1--19, December 1982.

\bibitem{vZ97}
{B.\ Van Zeghbroeck}, {\em Principles of semiconductor devices},
e-book published online at
\hspace*{1ex}\emph{ece-www.colorado.edu/~bart/book}, 1997.

\end{thebibliography}
\end{document}